\documentclass[submission,copyright,creativecommons]{eptcs}
\providecommand{\event}{TARK 2023} 

\usepackage{iftex}

\ifpdf
  \usepackage{underscore}         
  \usepackage[T1]{fontenc}        
\else
  \usepackage{breakurl}           
\fi

\title{Epistemic Logics of Structured Intensional Groups}
\author{Marta B\'ilkov\'a \qquad Igor Sedl\'ar\thanks{Work of both authors on this paper was supported by the grant no. 22-23022L CELIA of the Czech Science Foundation.}
\institute{Institute of Computer Science of the Czech Academy of Sciences\\ Prague, Czech Republic}
\email{\{bilkova, sedlar\}@cs.cas.cz}
}
\def\titlerunning{Epistemic Logics of Structured Intensional Groups}
\def\authorrunning{M.~B\'ilková \& I.~Sedl\'ar}

\usepackage{amsmath}
\usepackage{amsfonts}
\usepackage{amssymb}
\usepackage{amsthm}
\usepackage{stmaryrd}
\usepackage{bm}
\usepackage{textcomp} 
\usepackage{graphicx}
\usepackage[usenames,dvipsnames]{xcolor}
\usepackage{multicol}
\usepackage{tabularx}
\usepackage[normalem]{ulem}  

\usepackage{tikz}
\usetikzlibrary{arrows,positioning}

\newcommand{\f}{\varphi}
\newcommand{\ff}{\psi}
\newcommand{\fff}{\chi}

\newcommand{\p}{\alpha}
\newcommand{\pp}{\beta}
\newcommand{\ppp}{\gamma}

\newcommand{\leftd}{\scalebox{0.9}{\normalfont\texttt{\textlangle}}}
\newcommand{\rightd}{\scalebox{0.9}{\normalfont\texttt{\textrangle}}}
\newcommand{\leftb}{\scalebox{0.9}{\normalfont\texttt{[}}}
\newcommand{\rightb}{\scalebox{0.9}{\normalfont\texttt{]}}}
\newcommand{\B}[1]{\leftb #1 \rightb}  
\newcommand{\D}[1]{\leftd #1 \rightd}
\newcommand{\DB}[1]{\leftd #1 \rightb} 
\newcommand{\BD}[1]{\leftb #1 \rightd}

\newcommand{\tot}{\leftrightarrow}

\newcommand{\uf}[1]{G(#1)}

\newcommand{\blue}[1]{\textcolor{blue}{#1}}

\newtheorem{theorem}{Theorem}
\newtheorem{lemma}{Lemma}

\newtheorem{definition}{Definition}

\newtheorem*{claim*}{Claim}

	\theoremstyle{definition}
\newtheorem{example}{Example}
	\theoremstyle{remark}
\newtheorem{remark}{Remark}
\newtheorem*{proofsketch}{Proof sketch} 

\begin{document}
\maketitle

\begin{abstract}
Epistemic logics of intensional groups lift the assumption that membership in a group of agents is common knowledge. Instead of being represented directly as a set of agents, intensional groups are represented by a property that may change its extension from world to world. Several authors have considered versions of the intensional group framework where group-specifying properties are articulated using structured terms of a language, such as the language of Boolean algebras or of description logic. In this paper we formulate a general semantic framework for epistemic logics of structured intensional groups, develop the basic theory leading to completeness-via-canonicity results, and show that several frameworks presented in the literature correspond to special cases of the general framework.   
\end{abstract}

\section{Introduction}
One of the usual assumptions of multi-agent epistemic logic is that groups of agents are given \emph{extensionally} as sets of agents. Membership in an extensional group is common knowledge among all agents and change in membership implies change of identity of an extensional group. This is not how we usually think of groups, however. We are commonly reasoning about groups in various contexts without knowing their extensions---we might routinely refer to groups such as ``bot accounts'', ``democrats'', or ``correct processes''---and we do not settle for reducing such groups to their extensions either, as clearly they can change across the state space of a system, or possible states of the world. To reason about groups in a more realistic way is made possible by groups being given to us \emph{intensionally} by a common property.

In their seminal work \cite{grove1995naming,GroveHalpern1993}, Grove and Halpern introduced an elegant generalization of multi-agent epistemic logic where labels denoting (sets of) agents are replaced by abstract \emph{names} whose extensions can vary from world to world. Their language contains two types of modalities, equipped with a relational Kripke-style semantics: $E_n \f$ means that every agent in the current extension of $n$ knows that $\f$ (``everyone named $n$ knows''), and $S_n \f$ means that some agent in the current extension of $n$ knows that $\f$ (``someone named $n$ knows''). In the intensional setting, $S_n$ 
is in general not definable using disjunction and other epistemic operators. Grove and Halpern also consider a natural extension of their basic framework where names are replaced by formulas expressing \emph{structured} group-defining concepts. 

Motivated mainly by applications such as dynamic networks of processes, a framework where the agent set can vary not only across models, but also from state to state, have been developed in a form of term-modal logic TML. Introduced by \cite{Fittingetal2001TML}, TML builds upon first order logic, indexing modalities by terms that can be quantified over. TML is conveniently expressive but undecidable in general, and the attention therefore turns to identify some decidable fragments (\cite{padmanabhaLIPIcs2019,padmanabhaLIPIcs2018} (see \cite{shtakser2018propositional} for more references relevant for epistemic logic). Epistemic logic with names of \cite{GroveHalpern1993} was seminal in some sense to the development of TML, and can be seen as its simple decidable fragment. Another closely related language is that of implicitly quantified modal logic, studied in \cite{padmanabha2019propositional}. 

To model non-rigididity of group names, Kooi \cite{Kooi2007} introduces dynamic term-modal logic with assignment modalities. Wang and Seligman \cite{wang2018names} adopt a minimalist approach of using the basic assignment modalities with a quantifier-free term modal logic to obtain an easy-to-handle fragment of the logic in \cite{Kooi2007}, expressing various de re/de dicto distinctions in reading higher-order knowledge\footnote{Both ours and theirs formalisms implement term-indexed modalities in a two-sorted language, they are however languages with different expressive power---one algebraic, the other first-order---a precise comparison being a subject of future work.}.

Grove and Halpern's work is enjoying a recent resurgence of interest in the epistemic logic community. \cite{BilkovaEtAl2021} identifies a monotone neighborhood-style semantics for Grove and Halpern's language and, building directly on the $S_n$ and $E_n$ modalities, considers expansions with non-rigid versions of common and distributed knowledge.
Distributed or common knowledge for intensional group names has also been studied by \cite{moses1988programming,naumov2018everyone}. A monotone neighborhood perspective has recently been adopted by \cite{Dingetal2023} and applied to a logic containing the somebody-knows modality of \cite{AgotnesWang2021}.
Humml and Schr\"{o}der \cite{HummlSchroeder2023} generalize Grove and Halpern's approach to structured names represented by formulas defining group membership, including e.g.~formulas of the description logic ALC. Their abstract-group epistemic logic (AGEL) contains a common knowledge modality as the only modality and, unlike in \cite{BilkovaEtAl2021,GroveHalpern1993}, their group names are rigid.

In this paper, following \cite{BilkovaEtAl2021,GroveHalpern1993,padmanabha2019propositional} to various extent, we adopt the perspective that both ``everyone labeled $a$ knows'' and ``someone labeled $a$ knows'' form a minimal epistemic language for group knowledge, that groups are understood intensionally, and that labels reflect their structured nature. We use languages built on top of classical propositional language containing modalities $\B{a}$ and $\DB{a}$ indexed by elements of an algebra of a given signature of interest. 
As our main contribution, we set up a general framework for epistemic logics for structured groups in terms of relational semantics involving an algebra of group labels that index (sets of) relations in each world (Section~\ref{sec:relationalsemantics}), we show how some existing versions of frame semantics of closely related logics can be modelled in such a way, and then generalize relational frames in terms of two-sorted algebras involving propositions and groups, develop an algebraic duality and prove completeness of the minimal logic (Subsection~\ref{sec:algebraicduality}). We show that the semantics can be seen as an interesting version of monotone neighborhood frame semantics (Subsection~\ref{sec:nbhd}). In the remaining part of the paper we discuss several examples of algebraic signatures giving rise to interesting and useful variants of group structure (Section~\ref{sec:cases}).
%
\section{Frame semantics for structured groups}\label{sec:relationalsemantics}
%
%
\begin{definition}[\bfseries Relational frame]\label{def:relational}
Let $\Sigma$ be an algebraic similarity type. A $\Sigma$-algebra is any structure of the form $\mathbf{X} = (X, \{ o^{\mathbf{X}} \mid o \in \Sigma \})$, where each $o^{\mathbf{X}}$ is an $n$-ary operator on $X$ for some $n$.
A \emph{relational $\Sigma$-frame} is $\mathfrak{F} = (W, R, \mathbf{G})$, where $W \neq \emptyset$ (``worlds''); $R \subseteq 2^{W \times W}$ (``agent relations''); and $\mathbf{G}$ is a $\Sigma$-algebra with universe $G \subseteq (2^{R})^{W}$ (``group intensions'').
%
\end{definition}
In a relational $\Sigma$-frame, the set of available agents is represented by a set of accessibility relations $R$. Functions $f \in G$ map possible worlds $w \in W$ to sets $f(w) \subseteq R$ corresponding to sets of agents. These functions can be seen as \emph{intensions} of properties of agents: the intension $f$ of a given property determines for each world $w$ the \emph{extension} $f(w)$ of the property at $w$, representing the set of agents that possess the given property in $w$. Crucially, properties may change their extensions from world to world.
%
%
%
\begin{remark}
We note that a relational $\Sigma$-frame can be seen as a $\Sigma$-algebra over a subset of a direct product of a family of Kripke frames. In particular, $G \subseteq \prod_{w \in W} (W, Q_w)$ where $Q_w \subseteq R$; every $\mathfrak{F}$ then gives rise to $\mathfrak{G}(\mathfrak{F}) = (G, \{ o^{\mathfrak{G}(\mathfrak{F})} \mid o \in \Sigma \})$. Conversely, every $\mathfrak{G} = (G, \{ o^{\mathfrak{G}} \mid o \in \Sigma \})$ where $G \subseteq \prod_{w \in W} (W, Q_w)$ such that $Q_w \subseteq R$ fo all $w \in W$ gives rise to a relational $\Sigma$-frame. 
\end{remark}
%
%
\begin{definition}[\bfseries Language]\label{def:language}
Let $Pr,Gr$ be denumerable sets of propositional variables and group variables respectivelly. 
For each $\Sigma$, the \emph{$\Sigma$-language} is two-sorted, consisting of group $\Sigma$-terms and $\Sigma$-formulas. 
%
The set of $\Sigma$-terms $Tm_{\Sigma}$, and the set of $\Sigma$-formulas $Fm_{\Sigma}$, are defined by the following grammar\blue{s}:
$$ Tm_{\Sigma}: \quad \p :=  \mathtt{a}\in Gr \mid o (\p_1, \ldots, \p_n) \qquad
Fm_{\Sigma}: \quad \f := \mathtt{p}\in Pr \mid \neg \f \mid \f \land \f \mid \B{\p} \f \mid \DB{\p} \f.
$$
\end{definition} 
$\Sigma$-terms represent structured intensional groups where the structure is articulated using the operators of $\Sigma$ (number of examples follow). Formulas $\B{\p}\f$ read as ``Everyone in the group (given by) $\p$ believes that $\f$'' and $\DB{\p}\f$ read as ``Someone in the group (given by) $\p$ believes that $\f$''. We assume the standard definitions of Boolean operators ($\top, \bot, \lor, \to, \tot$), and we define $\D{\p}\f := \neg \B{\p} \neg \f$, $\BD{\p}\f := \neg \DB{\p} \neg \f$.
\begin{definition}[\bfseries Complex algebra]\label{def:complex}
The \emph{complex algebra} of $\mathfrak{F}$ is $\mathfrak{F}^{+} = (\mathbf{F}, \mathbf{G}, \B{\,}^{+}, \DB{\,}^{+})$ where
$\mathbf{F}$ is the Boolean algebra of (all) subsets of $W$ and
$\B{\,}^{+}$, $\DB{\,}^{+}$ are functions of the type $2^{W} \times \mathbf{G} \to 2^{W}$ such that for $a \in G$ and $P \subseteq W$: 
$$
\B{a}^{+}P = \{ w \mid \forall r \in a(w) : r(w) \subseteq P \}\qquad
\DB{a}^{+}P = \{ w \mid \exists r \in a(w) : r(w) \subseteq P \}
$$
(where $r(w) = \{ u \mid (w,u) \in r \}$).
%
\end{definition}
\begin{definition}[\bfseries Relational model]\label{def:model}
A \emph{model} based on a $\Sigma$-frame $\mathfrak{F} = (W, R, \mathbf{G})$ ($\Sigma$-model) is $\mathfrak{M} = (\mathfrak{F}, \llbracket\,\rrbracket)$, where $\llbracket \,\rrbracket$ (the ``interpretation function'') is a homomorphism from $Tm_{\Sigma} \cup Fm_{\Sigma}$ to $\mathfrak{F}^{+}$\footnote{Being a homomorphism, $\llbracket \cdot \rrbracket$ is determined by the values it assigns to variables.}, that is,
\begin{itemize}
\item $\llbracket \p\rrbracket \in \mathbf{G}$ where $\llbracket o (\p_1, \ldots, \p_n) \rrbracket = o^{\mathbf{G}} (\llbracket \p_1\rrbracket, \ldots, \llbracket \p_n\rrbracket)$;
\item $\llbracket \f \rrbracket \subseteq W$ where 
\begin{center}
$\llbracket \neg \f\rrbracket = W \setminus \llbracket \f\rrbracket $ \quad
$\llbracket \f \land \ff\rrbracket = \llbracket \f\rrbracket \cap \llbracket \ff\rrbracket$ \quad
$\llbracket \B{\p} \f\rrbracket = \B{\llbracket \p\rrbracket}^{+} \llbracket \f\rrbracket$ \quad
$\llbracket \DB{\p} \f\rrbracket = \DB{\llbracket \p\rrbracket}^{+} \llbracket \f\rrbracket$.
\end{center}
\end{itemize}
A formula $\f$ is \emph{valid} in a model $\mathfrak{M}$ iff $\llbracket \f \rrbracket_{\mathfrak{M}} = W_{\mathfrak{M}}$, and valid in a class of frames iff it is valid in each model based on a frame in the given class. If $\mathsf{K}$ is a class of frames, then $Log(\mathsf{K})$ is the set of formulas valid in all frames in $\mathsf{K}$.
\end{definition}
\begin{example}\label{ex:BCR21}
Consider a relational frame for epistemic logic with names \cite{BilkovaEtAl2021}. Let $N$ (``names''),  $A$ (``agents'') and $W$ (``worlds'') be three non-empty sets. A \emph{relational  frame} is $(W, A, N, Q, \mu)$, where $Q : A \to 2^{W \times W}$ and $\mu : N \to (W \to 2^{A})$. 
It is easy to see that each relational frame gives rise to a relational $\emptyset$-frame where $R = \{ Q_i \mid i \in A \}$ and $G = \{ \mu^{\#} (n) \mid n \in N \}$, where $\mu^{\#}(n)(w) = \{ Q_i \mid i \in \mu (n)(w)\}$. Conversely, every relational $\emptyset$-frame can be seen as a relational frame where $A = R$, $Q$ is the identity function on $A$, $N = G$ and $\mu (a)(w) = a(w)$ for all $a \in G$. 
\end{example}

\begin{example}\label{ex:GH93}
Grove and Halpern \cite{GroveHalpern1993} consider a version of their framework where groups are referred to by means of formulas of a Boolean language.
A simplified version of this framework can be presented as an extension of the relational frames of the previous example. 
In these frames we require that $N$ is a term algebra over terms in the signature  $\Sigma_{\mathsf{BA}} = \{ \,\bar{\,}, \land, \lor \}$, and that $\mu$ satisfies the following conditions (we use $n,m$ as variables ranging over $\Sigma_{\mathsf{BA}}$-term  to highlight the relation to Grove and Halpern's framework):
$$
\mu(\bar{n}, w) = W \setminus \mu(n, w) \qquad
\mu(n \land m, w) = \mu(n, w) \cap \mu(m, w)\qquad
\mu(n \lor m, w) = \mu(n, w) \cup \mu(m, w) \, .
$$
It is easy to see that every relational frame of this kind (Boolean relational frame) gives rise to a relational $\Sigma_{\mathsf{BA}}$-frame where $R$ and $G$ are defined as in the previous example. Conversely, every relational $\Sigma_{\mathsf{BA}}$-\emph{model} gives rise to a Boolean relational model: $A = R$, $Q$ is the identity function on $A$, $N$ is the term algebra over $\Sigma_{\mathsf{BA}}$-terms and $\mu (n) = \llbracket n\rrbracket$. The semantic clauses displayed above then follow from the assumption that the interpretation function $\llbracket\,\rrbracket$ is a homomorphism.
\end{example}



\begin{example}\label{ex:HS23}
In their recent work \cite{HummlSchroeder2023} on logic with common knowledge of abstract groups AGEL, Humml and Schr\"{o}der consider a rigid common knowledge operator for groups with membership defined by formulas. Technically, the common knowledge modality is
labeled by formulas in an agent language $\mathcal{L}_{Ag}$ built over a fixed set $Ag$ of agents, defining groups of agents by semantical means of an agent model $\mathit{A}$. A formula $C_\alpha \phi$ reads as $\phi$ is commonly known among agents satisfying $\alpha$. The language is interpreted over AGEL frames of the form $(W,\mathit{A},\sim)$ where $W$ is a set of worlds, and $\sim$ is a set of agent epistemic indistinguishability relations.
In the sense of this paper, their agent language $\mathcal{L}_{Ag}$ determines a signature $\Sigma$, and the complex algebra $\mathbf{A}$ of the agent model $\mathit{A}$, i.e. the algebra on group propositions $\{\llbracket\alpha\rrbracket_A\subseteq Ag \mid \alpha\in \mathcal{L}_{Ag}\} $, is a $\Sigma$-algebra. As the agent language conservatively extends classical propositional logic, this algebra carries a boolean structure. 
It gives rise to a $\Sigma$-relational frame where $R = \sim$ and the $\Sigma$-algebra $\mathbf{G}$ is determined by $\mathbf{A}$ on the universe consisting of assignments $g:W \to 2^{R}$, with $g(w) = \{\sim_{\llbracket\alpha\rrbracket_A} \mid \alpha\in \mathcal{L}_{Ag}\}$ for each $w\in W$, where $\sim_{\llbracket\alpha\rrbracket_A}$ is the union of relations of agents satisfying $\alpha$. 
\end{example}
\begin{remark}
    Our framework covers also semantics of modal logics with operations on accessibility relations. A prominent example are models for (test-free) Propositional Dynamic Logic. A relational PDL-model corresponds to a relational $\Sigma_{\mathsf{KA}}$-model, where $\Sigma_{\mathsf{KA}} = \{ \cdot, +, \,^{*}, 1, 0\}$ is the signature of Kleene algebra, such that $R$ is the set of all relations on a set of worlds $W$ and the functions in $G$ are constant and their values are singletons. In particular, $\mathbf{G}$ is the algebra of constant functions $f \in (2^R)^W$ such that $f(w)$ is a singleton (therefore we may identify $f$ with the $r$ such that $f(w) = \{ r \}$) and $f \cdot g$ is relational composition of $f$ and $g$, $f + g$ is the union of $f$ and $g$, $f^*$ is the reflexive transitive closure of $f$, $1$ is the identity relation, and $0$ is the empty relation.
\end{remark}

\begin{definition}[\bfseries Logic]\label{def:logic}
Let $\Sigma$ be an algebraic signature. An \emph{epistemic logic with structured intensional groups over $\Sigma$} (or simply a \emph{$\Sigma$-logic}) is any set $L \subseteq Fm_{\Sigma}$ such that (for all $\p \in Tm_{\Sigma}$)
\begin{enumerate}
\item $L$ contains all substitution instances of classical tautologies and is closed under Modus Ponens;
\item $L$ contains all formulas of the form $(\mathrm{K}) \, \B{\p} (\f \to \ff) \to (\B{\p}\f \to \B{\p}\ff)$ and is closed under the Necessitation rule $(\mathrm{Nec})\, \dfrac{\f}{\B{\p}\f}$;
\item  $L$ contains all formulas of the form $\neg \B{\p} \bot \to \DB{\p} \top$ and $\DB{\p}\f \land \B{\p}\ff \to \DB{\p}(\f \land \ff)$\footnote{It is easy to show that every $\Sigma$-logic contains all formulas of the form $\B{\p}\top \tot \top$ and $\B{\p}(\f \land \ff) \tot (\B{\p}\f \land \B{\p}\ff)$, and that it is closed under the rule $$ \dfrac{\f \land \ff_1 \land \ldots \land \ff_n \to \fff}{\DB{\f} \land \B{\p} \ff_1 \land \ldots \land \B{\p} \ff_n \to \DB{\p} \fff} \, .$$}.
\end{enumerate}
\end{definition}


\subsection{Algebraic duality}\label{sec:algebraicduality}
In this section we introduce specific two-sorted algebras that generalize relational $\Sigma$-frames. In a sense to be specified below, completeness results for classes of relational $\Sigma$-frames correspond to specific representation results for these two-sorted algebras.
\begin{definition}[\bfseries Frame]\label{def:frame}
Let $\Sigma$ be an algebraic similarity type. A \emph{$\Sigma$-frame} is $\mathfrak{A} = (\mathbf{F}, \mathbf{G}, \B{\,}, \DB{\,})$, where 
$\mathbf{F} = (X, \land, \lor, \neg, \top, \bot)$ is a Boolean algebra; $\mathbf{G} = (A, \{ o^{\mathbf{G}} \mid o \in \Sigma\})$ is a $\Sigma$-type algebra; and 
$\B{\,}$ and $\DB{\,}$ are functions of the type $\mathbf{F} \times \mathbf{G} \to \mathbf{F}$ such that
%
\begin{multicols}{2}
    \noindent
    \begin{align}
    \B{a} \top &= \top \label{ax1}\\
    \B{a}(x \land y) & = \B{a}x \land \B{a}y\label{ax2}        
    \end{align}
    \begin{align}
    \neg \B{a} \bot &\leq \DB{a} \top \label{ax3}\\
    \DB{a}x \land \B{a}y & \leq \DB{a}(x \land y) \label{ax4}
\end{align}    
\end{multicols}

%
\end{definition}
A $\Sigma$-frame is a two-sorted algebra bringing together a Boolean algebra of ``propositions'' with a $\Sigma$-algebra of ``groups''. The modal operators $\B{\,}$ and $\DB{\,}$, resembling scalar multiplication in modules, take pairs consisting of a group and a proposition to a proposition.
Formulas in $Fm_{\Sigma}$-can be seen as terms of the type corresponding to $\Sigma$-frames. In fact, we can define the following notion of an evaluation, leading to a natural definition of the equational theory of a class of $\Sigma$-frames.
\begin{definition}[\bfseries Equational theory]
An \emph{evaluation} on a $\Sigma$-frame is any homomorphism $Tm_{\Sigma} \cup Fm_{\sigma} \to \mathfrak{A}$, that is, any function $e$ such that
$e(o(\f_1, \ldots, \f_n)) = o^{\mathbf{F}} (e(\f_1), \ldots, e(\f_n))$ for all Boolean operators $o$;
 $e(o(\p_1, \ldots, \p_n)) = o^{\mathbf{G}} (e(\f_1), \ldots, e(\f_n))$ for all $\Sigma$-operators $o$;
 and $e(\B{\p}\f) = \B{e(\p)}e(\f)$ and $e(\DB{\p}\f) = \DB{e(\p)}e(\f)$.
%
A \emph{$\Sigma$-formula equation} is an expression of the form $\f \approx \ff$ where $\f, \ff \in Fm_{\Sigma}$. An equation $\f \approx \ff$ is \emph{valid in $\mathfrak{A}$} iff $e(\f) = e(\ff)$ for all evaluations $e$ on $\mathfrak{F}$. The \emph{equational theory} of a class $\mathsf{F}$ of $\Sigma$-frames is the set of all $\Sigma$-formula equations that are valid in all frames in $\mathsf{F}$, denoted as $Eq(\mathsf{F})$.\footnote{We note that it would make sense also to consider $\Sigma$-group equations as expressions of the form $\p \approx \pp$ where $\p, \pp \in Tm_{\Sigma}$, and define the group-equational theory of a class of frames, but we will not pursue this topic here.}.
\end{definition}  
\begin{remark}
    Dynamic algebras \cite{Kozen1980,Pratt1991a}, the algebraic counterparts of relational models for Propositional Dynamic Logic, are related to $\Sigma$-frames. A dynamic algebra is a pair $(\mathbf{F}, \mathbf{G}, \B{\,})$, where $\mathbf{F}$ is a Boolean algebra, $\mathbf{G}$ is a Kleene algebra, and $\B{\,} : \mathbf{G} \times \mathbf{F} \to \mathbf{F}$ satisfying our axioms (1--2) and further set of equations and quasi-equations. Therefore, dynamic algebras can be seen as a class of $\DB{\,}$-free reducts of $\Sigma_{\mathsf{KA}}$-frames.
\end{remark}

\begin{definition}[\bfseries Ultrafilter frame]\label{def:ultra}
Let $\mathfrak{A} = (\mathbf{F}, \mathbf{G}, \B{\,}, \DB{\,})$ be a $\Sigma$-frame. The \emph{ultrafilter frame of $\mathfrak{A}$} is $\mathfrak{A}_{+} = (\mathrm{Uf}(\mathbf{F}), R_{+}, \mathbf{G}_{+})$ where $\mathrm{Uf}(\mathbf{F})$ is the set of all ultrafilters on $\mathbf{F}$ (we define $\hat{x} = \{ u \in \mathrm{Uf}(\mathbf{F}) \mid x \in u \}$);
\begin{itemize}
\item $R_{+} = \{ r_{a, x} \mid x \in \mathbf{F} \And a \in \mathbf{G} \}$, where \ \ 
$r_{a, x} : w \mapsto \bigcap \{ \hat{y} \mid \B{a}y \in w\} \cap \hat{x} \, ;$
%
\item $G_{+} = \{ \uf{a} \mid a \in \mathbf{G} \} \subseteq (2^{R_{+}})^{\mathrm{Uf}(\mathbf{F})}$ such that $\forall u \in \mathrm{Uf}(\mathbf{F})$, \ $\uf{a}(u) = \{ r_{a, x} \mid \DB{a}x \in u \}$\\
(we will often write $\uf{a, u}$ instead of $\uf{a}(u)$);
%
\item $\mathbf{G}_{+} = (G_{+}, \{ o_{+} \mid o \in O_{\Sigma} \})$ where \ $o_{+} (\uf{a_1}, \ldots, \uf{a_n})(u) = \uf{o (a_1, \ldots, a_n)}(u)$.
%
\end{itemize}
\end{definition}
%
%
%
\begin{definition}[\bfseries Morphisms of $\Sigma$-frames]\label{def:morph}
Let $\mathfrak{A}_1, \mathfrak{A}_2$ be two $\Sigma$-frames. A ($\Sigma$-frame) \emph{morphism} is a function $f: \mathfrak{A}_1 \to \mathfrak{A}_2$ such that
\begin{multicols}{2}
\begin{itemize}
\item[(m1)] $f$ is a homomorphism from $\mathbf{F}_1$ to $\mathbf{F}_2$;
\item[(m2)] $f$ is a homomorphism from $\mathbf{G}_1$ to $\mathbf{G}_2$;
\item[(m3)] $f ( \B{a}_1 x) = \B{f(a)}_2 f(x)$; 
\item[(m4)] $f ( \DB{a}_1 x) = \DB{f(a)}_2 f(x)$.
\end{itemize}
\end{multicols}
\end{definition}
A \emph{quasi-embedding} of $\mathfrak{A}_1$ into $\mathfrak{A}_2$ is a morphism $f : \mathfrak{A}_1 \to \mathfrak{A}_2$ such that 
$f(x)=f(y) \to x=y$
for all $x,y$ in $\mathbf{F}_1$. An \emph{embedding} of $\mathfrak{A}_1$ into $\mathfrak{A}_2$ is a quasi-embedding where 
$f(a)=f(b) \to a=b$
for all $a,b$ in $\mathbf{G}_1$. A \emph{quasi-isomorphism} is a surjective quasi-embedding and an \emph{isomorphism} is a surjective embedding.
The \emph{canonical embedding algebra} of $\mathfrak{A}$ is $(\mathfrak{A}_{+})^{+}$ and the \emph{ultrafilter extension} of $\mathfrak{F}$ is $(\mathfrak{F}^{+})_{+}$. The \emph{canonical morphism} is a function $f: \mathfrak{A} \to (\mathfrak{A}_{+})^{+}$ with $f(x) = \hat{x}$ for $x \in \mathbf{F}$ and $f(a) = \uf{a}$ for $a \in \mathbf{G}$.
\begin{lemma}\label{lem:canmorph}
The canonical morphism is a quasi-embedding.
\end{lemma}
For each signature $\Sigma$, Lemma \ref{lem:canmorph} can be used to prove completeness of the basic $\Sigma$-logic with respect to all relational $\Sigma$-frames. In order to show this, we point out a useful example of a $\Sigma$-frame.

\begin{example}\label{ex:canframe-basic}
Let $L$ be a $\Sigma$-logic. Let $\equiv_L$ be a binary relation on $Fm_{\Sigma}$ such that $\f \equiv_L \ff$ iff $\f \tot \ff \in L$. Let $[\f]_L$ be the equivalence class of $\f$ under $\equiv_L$. It can be shown that $\equiv_L$ is a congruence on $Fm_{\Sigma}$. Hence, we obtain the Boolean algebra $\mathbf{F}^{L}$ of equivalence classes $[\f]_L$, where $o^{\mathbf{F}^{L}} ([\f_1]_L, \ldots, [\f_n]_L) = [o(\f_1, \ldots, \f_n)]_L$ for all Boolean operators $o$. We define $\mathbf{G}^{L}$ as the term algebra over $Tm_{\Sigma}$. Moreover, let $\B{\,}^{L}$ and $\DB{\,}^{L}$ be functions of the type $\mathbf{F}^{L} \times \mathbf{G}^{L} \to \mathbf{F}^{L}$ such that \ \ 
$ \B{\p}^{L} [\f]_L = [\B{\p}\f]_L \:\text{and}\: 
\DB{\p}^{L} [\f]_L = [\DB{\p}\f]_L
$
(note that these functions are well defined since $\equiv_L$ is a congruence).
Let us define the \emph{basic canonical $L$-frame} as $\mathfrak{B}^{L} = (\mathbf{F}^{L}, \mathbf{G}^{L}, \B{\,}^{L}, \DB{\,}^{L})$. It is clear that $\f \in L$ iff $\f \approx \top$ is valid in $\mathfrak{B}^{L}$.
\end{example}  

\begin{theorem}[\bfseries Completeness]\label{thm:completeness}
For all $\Sigma$, the smallest $\Sigma$-logic is the set of $\Sigma$-formulas valid in all relational $\Sigma$-models.
\end{theorem}
\begin{proof}
Fix a $\Sigma$ and take the smallest $\Sigma$-logic $L$. Soundness is easily checked. To show completeness, take the relational $\Sigma$-frame $(\mathfrak{B}^{L})_{+}$ (the canonical relational $L$-frame). Lemma \ref{lem:canmorph} entails that if $\f \notin L$, then $\f$ is not valid in $(\mathfrak{B}^{L})_{+}$. (Define a model where $\llbracket \f \rrbracket = \widehat{[\f]_L}$ and $\llbracket \p\rrbracket = \p$. Lemma \ref{lem:canmorph} implies that $\llbracket \,\rrbracket$ is indeed an interpretation function. Since $\f \tot \top \not\in L$, we have $\llbracket \f\rrbracket \neq \llbracket \top\rrbracket$ by the Prime Filter Theorem, and so $\f$ is not valid in $(\mathfrak{B}^{L})_{+}$.) 
\end{proof}

%
\subsection{Neighborhood semantics}\label{sec:nbhd}
The modalities $\DB$ and $\B$ are monotone modalities of the $\exists\forall$ and $\forall\forall$ type and can therefore be studied in terms of monotone neighborhood semantics, if we understand sets $\{r(w) \mid r\in a(w)\}$ as so called core neighborhood sets \cite{hansen2003monotonic,pacuit2017neighborhood}. 
Relational $\Sigma$-frames generalize relational frames for epistemic logic with names 
\cite{BilkovaEtAl2021,GroveHalpern1993} (Example~\ref{ex:BCR21}), which are categorialy equivalent to monotone neighborhood frames with neighborhood sets indexed by the set of names. 
Not surprisingly, a closely related connection arises between relational $\Sigma$-frames of this paper and monotone neighborhood frames where neighborhoods are indexed with algebraic terms. 
This will allow us to adapt and apply the well understood model theory of monotone neighborhood frames (for which we mainly refer to \cite{hansen2003monotonic,hansen2007bisimulation,pacuit2017neighborhood}) to study, among others, algebraic duality or modal definability on a convenient level of abstraction. A similar perspective has recently been adopted also by \cite{Dingetal2023} on a logic containing a somebody-knows modality, previously studied by \cite{AgotnesWang2021}. Neither of the approaches in \cite{hansen2003monotonic,Dingetal2023} however includes both $\exists\forall$ and $\forall\forall$ types of modalities, and therefore similar modifications of the general theory as those adopted in \cite{BilkovaEtAl2021} are necessary, and the algebraic structure underlying the labelling of groups needs to be captured additionally.
\begin{definition}[\bfseries Neighborhood frames]\label{def:nbhdframes}
A neighborhood $\Sigma$-frame $\mathfrak{F}$ is a tuple $(W, \mathbf{G}, \{\nu_a\}_{a \in G})$ where $W$ is a set of states, $\mathbf{G}$ is a $\Sigma$-algebra, and for each $a \in G$, 
$\nu_a: W \to 2^{2^{W}}$ 
is a neighborhood function that assigns to each state $w$ a set of sets of states\footnote{For the minimal $\Sigma$-logic, we do not require any additional (algebraic) properties from the assignment  $\nu_{\_}: \mathbf{G} \to [W,2^{2^{W}}]$. 
They might however become desirable in the examples that follow, and we will treat them as additional properties defining particular classes of frames (modally definable or not).}.
\end{definition}
%
%
\begin{definition}[\bfseries Semantics in neighborhood models]\label{def:satis:nbhd}
The complex algebra $\mathfrak{F}^+$ of a neighborhood $\Sigma$-frame $\mathfrak{F}$ is given as the expansion of the boolean algebra of subsets of $W$ by 
$$
\B{a}^{+}P = \{ w \mid \forall X \in \nu_a(w) : X \subseteq P \}\qquad
\DB{a}^{+}P = \{ w \mid \exists X \in \nu_a(w) : X \subseteq P \}.
$$
An interpretation function $\llbracket \rrbracket$ is a homomorphism from $Tm_{\Sigma} \cup Fm_{\Sigma}$ to $\mathfrak{F}^{+}$, i.e. 
\[
\llbracket \B{\p} \f \rrbracket = \{ w \mid \forall X\in \nu_{\llbracket\p \rrbracket}(w)  \ X\subseteq \llbracket\f\rrbracket\} \qquad
\llbracket \DB{\p} \f \rrbracket = \{ w \mid \exists X\in \nu_{\llbracket\p \rrbracket}(w)  \ X\subseteq \llbracket\f\rrbracket \}
\]
%
\end{definition}
%
%
\begin{definition}[\bfseries Neighborhood frame morphisms]\label{def:nbhdmorphisms}
Neighborhood $\Sigma$-frame morphisms are pairs of maps $(g:\mathbf{G} \to \mathbf{G'}, f:W \to W')$, where $g$ is a homomorphism of $\Sigma$-algebras, satisfying
\[
\mbox{(there)}\ X\in\nu_a(w) \Rightarrow f[X]\in\nu'_{g(a)}(f(w)) \qquad \mbox{(back)}\ Y\in\nu'_{g(a)}(f(w)) \Rightarrow \exists X(f[X]=Y\ \&\ X\in\nu_{a}(w))
\]
%
\end{definition}
Monotonicity of the modalities is built into the semantical definition rather than into the frame definition. As such, it corresponds to core neighborhood frames from \cite{hansen2003monotonic}, and the morphisms resemble core bounded morphisms of monotone neighborhood frames from \cite[Definition 4.6]{hansen2003monotonic}, additionally involving the algebraic homomorphism $g:\mathbf{G} \to \mathbf{G'}$ which can be interpreted as allowing to ``rename'' the groups along frame morphisms in a structured way. 
Understanding frame validity as $\mathfrak{F},w \Vdash \f$ if and only if $w\in \llbracket \f \rrbracket$ for each interpretation $\llbracket \rrbracket$ on $\mathfrak{F}$, we can prove that morphisms preserve frame validity:
\begin{lemma}[\bfseries Preservation of validity]\label{lem:nbhdmorphisms:preservation}
Let $(f,g): (W_1,\mathbf{G_1}, \{\nu_a\}_{a \in G_1}) \to (W_2,\mathbf{G_2}, \{\nu_a\}_{a \in G_2})$ be a neighborhood $\Sigma$-frame morphism from $\mathfrak{F}_1$ to $\mathfrak{F}_2$. Then for each formula $\f$ and each $w\in W$, 
$$
\mathfrak{F}_1,w \Vdash \f \ \ \ \Rightarrow \ \ \mathfrak{F}_2,f(w) \Vdash \f.
$$
\end{lemma}
%
%
A proof-sketch can be found in the Appendix~\ref{ap:preservation}. For the sake of interest we also spell out in Appendix~\ref{ap:bisimulations} what bisimulations of neighborhood $\Sigma$-frames look like.

For a relational $\Sigma$-frame $\mathfrak{F} = (W,R,\mathbf{G})$, we can define the corresponding neighborhood $\Sigma$-frame $\mathfrak{F}^n = (W,\mathbf{G},\{\nu^n_a\}_{a \in G})$ putting 
$
\nu^n_a(w) = \{r(w) \mid r\in a(w)\}. 
$
Conversely, for a neighborhood $\Sigma$-frame $\mathfrak{F} = (W,\mathbf{G},\{\nu_a\}_{a \in G})$ we define the corresponding relational $\Sigma$-frame $\mathfrak{F}^r = (W,R^r,\mathbf{G})$ by 
$
a^r(w) = \{r \mid r(w)\in \nu_a(w)\}, \ R^r = \bigcup_{a\in G,w\in W} a^r(w). 
$
We then obtain the following: 
\begin{theorem}[\bfseries Categorial equivalence]\label{thm:categorialequivalence}
The categories of relational $\Sigma$-frames and neighborhood $\Sigma$-frames are equivalent. 
\end{theorem}
Given the completeness of the basic $\Sigma$-logic with respect to relational $\Sigma$-frames (Theorem~\ref{thm:completeness}), completeness with respect to all neighborhood $\Sigma$-frames follows\footnote{It is also possible to define a canonical neighborhood $\Sigma$-frame directly, following similar pattern as in Definition~\ref{def:ultra}.}
With the complex algebra/ultrafilter frame construction at hand, we can describe the algebraic duality, and obtain a definability theorem characterizing modally definable classes of neighborhood $\Sigma$-frames (cf. Theorem 2 of \cite{BilkovaEtAl2021}). 

\section{Special cases}\label{sec:cases}
To illustrate some interesting special cases of the general framework discussed above, we introduce, in each case, a class of relational frames that captures some natural kind of structure imposed on intensional groups, provide an algebraic generalization of relational frames, and show that the respective classes of relational frames and their algebraic generalizations determine the same logic. 
\subsection{Unions and join-semilattices}\label{sec:cases-JS}
One of the simplest forms of structure imposed on groups of agents corresponds to taking unions of sets of agents. On the intensional perspective, taking unions corresponds to an operation on intensional groups that, for each world $w$, gives the union of the extensions of the given intensional groups in $w$. It is then natural to impose a \emph{semilattice} structure on the set of intensional groups, where the neutral element is an ``inconsistent'' intensional group that has an empty extension in each world. This case is also easily handled in the technical sense, and so we will discuss it as an introductory example. 

Nevertheless, even this simple case has an interesting feature: the ultrafilter frame construction does not in general lead to a relational frame of the right kind. This may be surprising given the fact that unions are well-behaved in the extensional framework. This feature is discussed at the end of the section. 

\begin{definition}[\bfseries JS-frame]\label{def:JS}
Let $\Sigma_{\mathsf{SL}} = \{ +, 0 \}$ be the join-semilattice signature. A \emph{relational join-semi\-lattice frame} (relational js-frame) is a relational $\Sigma_{\mathsf{SL}}$-frame where $0^{\mathbf{G}}(w) = \emptyset$ and $(f +^{\mathbf{G}} g)(w) = f(w) \cup g(w)$.
%
%
A \emph{join-semilattice frame} (js-frame) is a $\Sigma_{\mathsf{SL}}$-frame where $\mathbf{G}$ is a join semilattice and 
%
\begin{multicols}{2}
    \noindent
    \begin{align}
    \B{0}x & = \top \label{js1}\\
    \DB{0}x & = \bot \label{js2}        
    \end{align}
    \begin{align}
    \B{a + b}x &= \B{a}x \land \B{b}x \label{js3}\\
    \DB{a + b}x &= \DB{a}x \lor \DB{b}x \label{js4}
\end{align}    
\end{multicols}
The class of (relational) js-frames will be denoted as $\mathsf{FSL}$ ($\mathsf{rFSL}$). 
\end{definition}
\begin{definition}[\bfseries JS-logic]
The \emph{join-semilattice logic} $LSL$ is the smallest $\Sigma_{\mathsf{SL}}$-logic that contains all formulas of the following forms:
%
\begin{multicols}{2}
    \noindent
    \begin{align}
    \top & \to \B{0}\f\tag{a\ref{js1}} \label{ax-js1}\\
    \DB{0}\f & \to \bot\tag{a\ref{js2}} \label{ax-js2}        
    \end{align}
    \begin{align}
    \B{\p + \pp}\f & \tot \B{\p}\f \land \B{\pp}\f\tag{a\ref{js3}}\label{ax-js3}\\
    \DB{\p + \pp}\f & \tot \DB{\p}\f \lor \DB{\pp}\f\tag{a\ref{js4}} \label{ax-js4}
\end{align}    
\end{multicols}
\end{definition}
\begin{theorem}\label{thm:complet-js}
(1) $\f \in LSL$\ \ iff\ \ (2) $\f \in Log(\mathsf{rFSL})$\ \ iff\ \ (3) $ (\top \approx \f) \in Eq(\mathsf{FSL})$.
\end{theorem}
\begin{proofsketch}
(1) implies (2) since the $LSL$ axioms are valid in all relational js-frames. The fact that (2) implies (3) is established by showing that for each js-frame there is an equivalent relational js-frame. We cannot use the ultrafilter frame construction (Def.~\ref{def:ultra}) for failure of canonicity\footnote{Axiom \ref{ax-js4} $\DB{\p + \pp}\f \tot \DB{\p}\f \lor \DB{\pp}\f$ does not correspond to the condition $(f +^{\mathbf{G}} g)(w) = f(w) \cup g(w)$, but to the following one:
$
\forall w \ (\forall r\in (a+b)(w) \ \exists s\in a(w)\cup b(w) \ (s(w)\subseteq r(w))) \wedge \forall w \ (\forall s\in a(w)\cup b(w)  \ \exists r\in (a+b)(w) \ (r(w)\subseteq s(w))).
$
While the first conjunct is valid on an ultrafilter frame, the second one is not (unless we deal with an ultrafilter frame of a complete algebra).
For the $\Sigma_{\mathsf{SL}}$-neighborhood frames, \ref{ax-js4} corresponds to the property:
$
\forall w (\nu_{a+b}(w)^{\uparrow} = \nu_{a}(w)^{\uparrow} \cup \nu_{b}(w)^{\uparrow}).
$}. However, a variant of the construction where $0$ and $+$ are defined exactly as in relational js-frames will do. That (3) implies (1) is established by contraposition, using a variant of the basic canonical $L$-frame of Example \ref{ex:canframe-basic}. Details are given in Appendix \ref{ap:semilattices}.
\qed
%
\end{proofsketch}
%
\subsection{Meta-belief and right-unital magmas}\label{sec:cases-M}

Information about \emph{meta-beliefs} (``$i$ believes that $j$ believes that $p$'') is crucial to many multi-agent scenarios. The notion of meta-belief is often lifted to extensional groups of agents (sets). ``Group $I$ believes that group $J$ believes that $p$'' means that every agent in $I$ believes that every agent in $J$ believes that $p$. It is interesting to note that, if agents are seen as accessibility relations, the notion of meta-belief induces structure on sets of agents. In particular, every agent in $I$ believes that every agent in $J$ believes that $p$, iff every world accessible via $I \circ J = \{ r \circ q \mid r \in I \And q \in J \}$ satisfies $p$. If the ``environment'' agent $E = \{ \mathrm{id}_W \}$ is also included, we obtain a monoid structure. It is interesting to look at the notion of meta-belief, and the structure it induces, in the context of intensional groups. 

\begin{example}\label{ex:comp}
Adam ($\mathcal{A}$) is reviewing a paper for a journal, double-blind. Adam knows the researchers active in the particular area very well, and so he knows that either Bonnie ($\mathcal{B}$) or Carrie ($\mathcal{C}$) is the author or they are co-authoring the paper together. He knows that the authors of the paper, whoever they are, believe that the proof of a particular statement in the paper is correct ($p$), although Adam believes it is incorrect ($\neg p$). In reality, Bonnie and Carrie co-authored the paper and the proof is correct. 

The scenario is represented by the relational model in Figure \ref{fig:comp}, with the actual world underlined. Adam's meta-beliefs concerning the authors of the paper are represented by the result of composing his relation with relations that ``behave like'' a relation corresponding to an author of the paper in any world accessible for Adam from the actual world. In particular, from the world $(\{ \mathcal{C} \}, \neg p)$, representing the situation where only Carrie is the author and the proof is incorrect, only the $\mathcal{C}$-arrow is followed, and similarly for the world $(\{ \mathcal{B} \}, \neg p)$ and the $\mathcal{B}$ arrow. This makes sense: beliefs of people who are not authors in the given world are disregarded. In the world $(\{ \mathcal{B}, \mathcal{C} \}, \neg p)$, one could follow either $\mathcal{B}$ or $\mathcal{C}$, but the difference is not reflected by the accessibility arrows leading from that world.
\end{example}

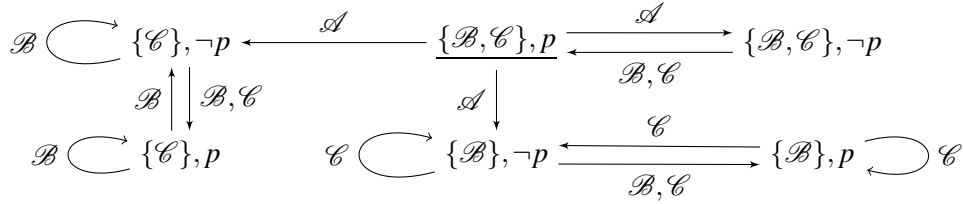
\begin{figure}\centering
\begin{tikzpicture}
\node at (-4.2,0) (u0) {$\{ \mathcal{C} \}, p$};
\node at (0,0) (u1) {$\{ \mathcal{B} \}, \neg p$};
\node at (4.2,0) (u2) {$\{ \mathcal{B} \}, p$};
\node at (-4.2,1.5) (w0) {$\{ \mathcal{C} \}, \neg p$};
\node at (0,1.5) (w1) {\underline{$\{ \mathcal{B}, \mathcal{C} \}, p$}};
\node at (4.2,1.5) (w2) {$\{ \mathcal{B}, \mathcal{C} \}, \neg p$};

\path[-latex']
(w1) edge node[left] {$\mathcal{A}$} (u1)
(w1.8) edge node[above] {$\mathcal{A}$} (w2.173)
(w1) edge node[above] {$\mathcal{A}$} (w0)
(w2.-173) edge node[below] {$\mathcal{B}, \mathcal{C}$} (w1.-8)
(u1.-8) edge node[below] {$\mathcal{B}, \mathcal{C}$} (u2.-170)
(u2.171) edge node[above] {$\mathcal{C}$} (u1.9)
(u1) edge[loop left] node[left] {$\mathcal{C}$} (u1)
(u2) edge[loop right] node[right] {$\mathcal{C}$} (u2)
(w0.-70) edge node[right] {$\mathcal{B}, \mathcal{C}$} (u0.70)
(u0.110) edge node[left] {$\mathcal{B}$} (w0.-110)
(u0) edge[loop left] node[left] {$\mathcal{B}$} (u0)
(w0) edge[loop left] node[left] {$\mathcal{B}$} (w0);
\end{tikzpicture}\caption{A relational model corresponding to Example \ref{ex:comp}.}\label{fig:comp}
\end{figure}

%

Let $W$ be a set and let $R \subseteq (2^{W})^{W}$. 
 Let $f \in (2^{R})^{W}$ be an intensional group. A \emph{variant} of $f$ is a relation $r \in (2^{W})^{W}$ such that, for each $w \in W$, if $f(w) \neq \emptyset$, then there is $q \in f(w)$ such that $r(w) = q(w)$, and if $f(w) = \emptyset$, then $r(w) = \emptyset$. We denote the set of all variants of $f$ as $f^{\dagger}$. 

Intuitively, a variant of an intensional group $f$ is a relation that ``behaves like'' some relation in $f(w)$ whenever $f(w)$ is non-empty (not necessarily the same relation!) and which is ``blind'' in $w$ where $f(w)$ is empty. 

%

\begin{definition}[\bfseries Intensional composition]
Let $R$ be a set of binary relations on a set $W$. We define the operation $\otimes : (2^{R})^{W} \times (2^{R})^{W} \to (2^{R})^{W}$ point-wise by $(f \otimes g)(w) = f(w) \circ g^{\dagger}$.
\end{definition}
It can be shown that the structure on the set of intensional groups induced by intensional composition is rather weak. For instance, the natural candidate for the unit element, the function $1$ that maps each $w$ to $\{ \mathrm{id}_W \}$ is the right unit, but not the left unit: $(f \otimes 1)(w) = f(w) \circ 1^{\dagger} = f(w) \circ \{ \mathrm{id} _W\} = f(w)$, but in general $(1 \otimes f)(w) = 1(w) \circ f^{\dagger} = f^{\dagger} \neq f(w)$. (We note that in general $f^{\dagger} \neq f(w)$ since we can have $r \in f(w)$ and $r(v) \neq \emptyset$ for some $v \neq w$ such that $f(v) = \emptyset$.)

A \emph{right-unital magma} (rum) is an algebra $(M, \cdot, 1)$ where $\cdot$ is a binary operation on $M$ and $1 \in M$ such that $x \cdot 1 = x$ for all $x \in M$.

\begin{definition}[\bfseries Rum-frame]
Let $\Sigma_{\mathsf{M}} = \{ \cdot, 1 \}$ be the monoid signature. A \emph{relational right-unital-magma frame} (a relational rum-frame) is a relational $\Sigma_{\mathsf{M}}$-frame such that $1^{\mathbf{G}}(w) = \{ \mathrm{id}_W \}$ and $(g \cdot^{\mathbf{G}} h)(w) = (g \otimes h)(w)$.
A \emph{rum-frame} is a $\Sigma_\mathsf{M}$-frame where $\mathbf{G}$ is a right-unital-magma and
\begin{multicols}{2}
    \noindent
    \begin{align}
    \B{1}x & = x \label{mon1}\\
    \DB{1} x & = x\label{mon2}        
    \end{align}
    \begin{align}
    \B{a \cdot b}x & = \B{a} \B{b}x \label{mon3}\\
    \DB{a \cdot b} x & = \DB{a} (\BD{b}\bot \lor \DB{b}x) \label{mon4}
\end{align}    
\end{multicols}
\end{definition}

The class of (relational) rum-frames will be denoted as $\mathsf{FRUM}$ ($\mathsf{rFRUM}$, respectively).


\begin{definition}[\bfseries Rum-logic]
The \emph{right-unital-magma logic} $LRUM$ is the smallest $\Sigma_{\mathsf{M}}$-logic that contains all formulas of the following forms:
\begin{multicols}{2}
\noindent
\begin{align}
\B{1} \f & \tot \f\label{mon1-ax}\\
\DB{1}\f & \tot \f\label{mon2-ax}
\end{align}
\begin{align}
\B{\p \cdot \pp}\f & \tot \B{\p}\B{\pp}\f\label{mon3-ax}\\
\DB{\p \cdot \pp}\f & \tot \DB{\p}(\BD{\pp}\bot \lor \DB{\pp}\f)\label{mon4-ax}
\end{align}    
\end{multicols}

\end{definition}

\begin{remark}\label{rem:composition-axiom}
We note that a simpler variant of (\ref{mon3-ax}) for $\DB{\,}$, namely, $\DB{\p \cdot \pp}\f \tot \DB{\p}\DB{\pp}\f$ is not valid. In particular, the left-to-right implication has the following counterexample. Let $\llbracket \p\rrbracket(w) = \{ r \}$ and let $r(w) = \{ u, v \}$; moreover, let us assume that $\llbracket \pp\rrbracket(u) = \{ q \}$, $\llbracket \pp\rrbracket(v) = \emptyset$, $q(u) = \{ u \} = \llbracket p\rrbracket$. It is easily checked that $\llbracket \p \cdot\pp\rrbracket(w) = \{ \{ (w, u) \} \}$, and so $w \models \DB{\p \cdot \pp} p$. However, $w \not\models \DB{\p}\DB{\pp}p$, since this would require $\llbracket \pp\rrbracket(v)$ to be non-empty.  On the other hand, (\ref{mon4-ax}) is valid since worlds $v$ accessible via $\p$ where $\llbracket \pp\rrbracket(v) = \emptyset$ are taken care of by the extra disjunct $\BD{\pp}\bot$. 
\end{remark}

\begin{theorem}\label{thm:complet-mon}
$\f \in LRUM$ iff $\f \in Log(\mathsf{rFRUM})$ iff $ (\top \approx \f) \in Eq(\mathsf{FRUM})$. 
\end{theorem}

\subsection{Closure semilattices and distributed knowledge}\label{sec:cases-CS}

In the extensional setting, $\f$ is distributed knowledge in a group iff it is satisfied in every world accessible using the intersection of the relations in the group. If all relations are reflexive, then $\f$ is distributed knowledge iff there is a non-empty subset of the given group such that $\f$ is satisfied in all worlds accessible using the intersection. On the relations-as-agents perspective, the intersection of each non-empty subset of a set of relations-agents gives rise to a new relation-agent. Hence, forming intersections of non-empty subsets of a group $X$ transforms the group of relations-agents into a new group $X'$. Interestingly, distributed knowledge in $X$ then corresponds to the ``somebody knows'' operator applied to $X'$. Hence, distributed knowledge induces structure on groups of agents even in the extensional setting. We will look at the structure induced by distributed knowledge in the intensional setting.  

\begin{definition}[CSL-frame]
Let $\Sigma_{\mathsf{CSL}} = \{ +, 0, \,^{\cap} \}$ where $\{ +, 0 \}$ is the join-semilattice signature and $\,^{\cap}$ is a unary operator. A \emph{relational closure semilattice frame} (relational cs-frame) is a relational $\Sigma_{\mathsf{CSL}}$-frame where all $r \in R$ are reflexive and $0^{\mathbf{G}}(w) = \emptyset$, $(f +^{\mathbf{G}} g)(w) = f(w) \cup g(w)$, and  
$$ f^{\cap^{\mathbf{G}}}(w) = \{ r \in R \mid r(w) = \bigcap_{r_i \in X}  r_i(w)  \text{ for some } \emptyset \neq X \subseteq f(w) \}$$\footnote{We assume that $r \in R$ in the frame is closed under this operation.}
A \emph{closure semilattice frame} (cs-frame) is a $\Sigma_{\mathsf{CSL}}$-frame where $\mathbf{G}$ is a join semilattice with partial order defined as usual ($a \leq b$ iff $a + b = b$), $\,^{\cap}$ is a closure operator on $\mathbf{G}$, the join-semilattice axioms (\ref{js1}--\ref{js4}) are satisfied as well as $0^\cap = 0$, and 
\begin{multicols}{2}
    \noindent
    \begin{align}
    \B{a}x & \leq x\label{ref}\\
    \DB{a^{\cap}}x \land \DB{a^{\cap}}y  & \leq \DB{a^{\cap}}(x \land y)\label{cs2}        
    \end{align}
    \begin{align}
    \B{a^{\cap}}x & = \B{a}x\label{cs3}\\
    \DB{a^{\cap}}x & \leq \DB{a}\top\label{cs4}        
    \end{align}
\end{multicols}
The class of (relational) cs-frames is denoted as $\mathsf{FCS}$ ($\mathsf{rFCS}$).
\end{definition}
%
%
\begin{definition}[\bfseries CS-logic]
The \emph{closure semilattice logic} $LCS$ is the smallest $\Sigma_{\mathsf{CSL}}$-logic that extends $LSL$ and contains all formulas of the following forms: 
\begin{multicols}{2}
  \noindent
  \begin{align}
    \B{\p}\f & \to \f\label{ref-ax}\\
    \DB{\p^{\cap}}\f \land \DB{\p^{\cap}}\ff & \to \DB{\p^{\cap}}(\f \land \ff)\label{cs2-ax}\\
    \DB{\p}\f & \to \DB{\p^{\cap}}\f \label{cs5-ax}
  \end{align}
  \begin{align}
    \B{\p^{\cap}}\f & \tot \B{\p}\f\label{cs3-ax}\\
    \DB{\p^{\cap}}\f & \to \DB{\p}\top\label{cs4-ax}\\
    \DB{\p^{\cap\cap}}\f & \to \DB{\p^{\cap}}\f \label{cs6-ax}
\end{align}
\end{multicols}
and closed under the rule 
$$\dfrac{\DB{\p}\f \to \DB{\pp}\f}{\DB{\p^{\cap}}\f \to \DB{\pp^{\cap}}\f}.$$

\end{definition}

\begin{theorem}\label{thm:complet-cs}
$\f \in LCS$ iff $\f \in Log(\mathsf{rFCS})$ iff $ (\top \approx \f) \in Eq(\mathsf{FCS})$.
\end{theorem}

\section{Further work}
With a reasonable notion of composition of intensional groups, we may use the standard fixpoint construction to introduce common knowledge into our framework. We intend to study the extension with common knowledge in the immediate future. An additional topic for future work is the exploration of variants of the notion of intensional composition. In particular, we are curious if there is a variant giving rise to a monoid structure on intensional groups.
\paragraph{Acknowledgement.} The authors are grateful to the reviewers for useful comments. A preliminary version of this paper was presented at the CELIA Project Meeting in Bayreuth, Germany on 23-24 February 2023. We acknowledge the valuable suggestions we received from the audience at the meeting.

\bibliographystyle{eptcs}
\bibliography{celia}

\newpage
\appendix
\section{Appendix}

\subsection{Proof of Lemma~\ref{lem:canmorph}}

\begin{proof}
It is a standard observation that $f$ embeds any Boolean algebra $\mathbf{F}$ into the power set algebra over $\mathrm{Uf}(\mathbf{F})$. This takes care of (m1) and the injectivity condition. (m2) holds by definition. 

(m3) is established as follows (we write $\B{\:}$ instead of $\B{\:}^{+}$). The inclusion $f(\B{a}x) \subseteq \B{\uf{a}} \hat{x}$ means that if $\B{a}x \in u$ and $r \in \uf{a,u}$ for some $u$, then $r(u) \subseteq \hat{x}$. This holds by the definition of $\uf{a}$.   The converse inclusion $\B{\uf{a}} \hat{x} \subseteq f(\B{a}x)$ is established by contraposition. Assume that $u \notin f(\B{a}x)$, that is, $\B{a}x \notin u$ for some $a$ and $x$. Using (\ref{ax1}--\ref{ax2}), we can show that $v_0 = \{ y \mid \B{a}y \in u \}$ is a filter on $\mathbf{F}$ such that $x \notin v_0$. Hence, $v_0$ extends to an ultrafilter $v$ such that $x \notin v$. Take the relation $r_{a, \top}$, where $\top := x \lor \neg x$ for some $x \in \mathbf{F}$. By (\ref{ax2}--\ref{ax3}), $\B{a}x \notin u$ implies $\DB{a}\top \in u$, and so $r_{a,\top} \in \uf{a, u}$. Moreover, $r_{a,\top}(u,v)$ by the construction of $v$. This means that $u \notin \B{\uf{a}} \hat{x}$.

(m4) is established as follows. The inclusion $f(\DB{a}x) \subseteq \DB{\uf{a}} \hat{x}$ means that if $\DB{a}x \in u$ for some $u$, then there is $r \in \uf{a,u}$ such that $r(u) \subseteq \hat{x}$. Fix such $a, x$ and $u$, and consider the relation $r_{a,x}$. It is clear that $r_{a,x} \in \uf{a, u}$ and $r_{a,x}(u) \subseteq \hat{x}$. 
The converse inclusion $\DB{\uf{a}} \hat{x} \subseteq f(\DB{a}x)$ is established as follows. Let us assume that $\DB{\uf{a}} \hat{x}$, i.e.~there is $r \in \uf{a, u}$ such that $r(u) \subseteq \hat{x}$. This means that $r = r_{a,z}$ for some $z$ such that $\DB{a}z \in u$, and $r(u) = \bigcap \{ \hat{y} \mid \B{a}y \in u \}\cap \hat{z}$. This entails that
$$ \{ z \} \cup \{ y \mid \B{a}y \in u \} \subseteq w \implies x \in w$$ for all $w \in \mathrm{Uf}(\mathbf{F})$. Hence, $x$ is in the filter generated by  $\{ z \} \cup \{ y \mid \B{a}y \in u \}$. By the properties of filters generated by (non-empty) subsets of a lattice, there is a finite $\{ y_1, \ldots, y_n \} \subseteq \{ y \mid \B{a}y \in u \}$ such that $ z \land y_1 \land \ldots \land y_n \leq x$. This means that
$$ \DB{a}z \land \B{a}y_1 \land \ldots \land \B{a}y_n \leq \DB{a}x, $$ using (\ref{ax2}) and (\ref{ax4}). Consequently, $\DB{a}x \in u$ as we wanted to show.
\end{proof}

\subsection{Proof of Lemma~\ref{lem:nbhdmorphisms:preservation} (Preservation of validity for neighborhood frame morphisms)}\label{ap:preservation}

\begin{proof}
To show for each formula $\f$ and each $w\in W$, 
$
(\mathfrak{F}_1,w \Vdash \f \ \ \ \Rightarrow \ \ \mathfrak{F}_2,f(w) \Vdash \f),
$
 it is enough to show that, once we fix valuations on the respective frames so that (i) for each $\mathtt{a} \in Gr$, $g(\llbracket \mathtt{a} \rrbracket_1) = \llbracket \mathtt{a} \rrbracket_2$, and (ii) for each $\mathtt{p} \in Pr$, $w\in \llbracket \mathtt{p} \rrbracket_1$ iff $f(w)\in \llbracket \mathtt{p} \rrbracket_2$, we obtain for each $\f$
$$
w\in \llbracket \f \rrbracket_1 \ \ \Leftrightarrow \ \ f(w)\in  \llbracket \f \rrbracket_2.
$$
This is easily proven by a routine induction on the complexity of a given formula.
\end{proof}

\subsection{Bisimulations of neighborhood frames}\label{ap:bisimulations}
To see what a natural notion of bisimulation is for neighborhood $\Sigma$-models, compared to that of models for epistemic logic with names described in \cite[Definition 6]{BilkovaEtAl2021}, we need to incorporate the algebraic component. For a binary relation $B$, let $X\mathrel{\overline{B}} Y$ iff $\forall x\in X \exists y\in Y\ xBy$ and $\forall y\in Y \exists x\in X \ xBy$.
\begin{definition}[\bfseries Bisimulations]\label{def:nbhd:bisimulation}
Let $(W_1,\nu^1,\mathbf{G_1},\llbracket \rrbracket_1)$ and $(W_2,\nu^2,\mathbf{G_2},\llbracket \rrbracket_2)$ be neighborhood $\Sigma$-models.
A pair $(\cong,B)$, with $\cong \subseteq \mathbf{G_1}\times\mathbf{G_2}$ being a congruence relation, and $B\subseteq W_1\times W_2$, is a bisimulation of neighborhood $\Sigma$-models, if 
\begin{itemize}
    \item[] $\forall\mathtt{a}\in Gr \  \llbracket \mathtt{a} \rrbracket_1 \cong  \llbracket \mathtt{a} \rrbracket_2$ \ and \ $\forall \mathtt{p}\in Pr \ (w_1 B w_2 \Rightarrow (w_1\in \llbracket \mathtt{p} \rrbracket_1 \Leftrightarrow w_2\in \llbracket \mathtt{p} \rrbracket_2))$
    \item[] $(w_1 B w_2 \wedge a_1\cong a_2) \Rightarrow (\forall X\in\nu^1_{a_1}(w_1) \exists Y\in \nu^2_{a_2}(w_2)\  X\mathrel{\overline{B}} Y ) \wedge (\forall Y\in \nu^2_{a_2}(w_2) \exists X\in\nu^1_{a_1}(w_1)\  X\mathrel{\overline{B}} Y)$
\end{itemize}
\end{definition}
As expected, bisimilarity implies modal equivalence for the language of the basic $\Sigma$-logic, and the converse holds for image-finite models (where every core neighborhood set is a finite set of finite sets). Graphs of neighborhood $\Sigma$-frame morphisms are prominent examples of bisimulations, and functional bisimulations correspond to graphs of neighborhood $\Sigma$-frame morphisms.

\subsection{Proof of Theorem~\ref{thm:categorialequivalence} (Categorial equivalence)}

\begin{proof}
For a relational $\Sigma$-frame $\mathfrak{F} = (W,R,\mathbf{G})$, we define the corresponding neighborhood $\Sigma$-frame $\mathfrak{F}^n = (W,\mathbf{G},\{\nu^n_a\}_{a \in G})$ as follows:
$$
\nu^n_a(w) = \{r(w) \mid r\in a(w)\}.
$$
Conversely, for a neighborhood $\Sigma$-frame $\mathfrak{F} = (W,\mathbf{G},\{\nu_a\}_{a \in G})$ we define the corresponding  $\Sigma$-frame $\mathfrak{F}^r = (W,R^r,\mathbf{G})$ as follows:
$$
a^r(w) = \{r \mid r(w)\in \nu_a(w)\}, \ R^r = \bigcup_{a\in G,w\in W} a^r(w).
$$
It is easy to see that $\nu^{n}_a (w) = \{r(w) \mid r\in a^r(w)\} = \nu_a (w)$ and $a^{r}(w) = \{r \mid r(w)\in \nu^n_a(w)\} = a(w)$. However, going there-and-back on a relational frame, we do not recover the same $R$, as we can in principle recover only those relations $r$ in $R$ that are in some $a(w)$, but we also include relations who agree with $r$ on $w$. Still the resulting frame ends up to be isomorphic to the original one in terms of frame morphisms, which is what matters.

For the morphism part, we use the fact that the corresponding frames are defined over the same set $W$ and $\Sigma$-algebra $\mathbf{G}$, and therefore we use the same underlying map in both directions. First, we observe that morphisms of relational $\Sigma$-frames (which can be read of the Definition \ref{def:morph} of morphisms of $\Sigma$-frames) can equivalently be understood as pairs of maps $(g:\mathbf{G} \to \mathbf{G'}, f:W \to W')$, where $g$ is a homomorphism of $\Sigma$-algebras, satisfying:
\begin{itemize}
\item[] (there) $\forall r\in a(w) \  \exists r'\in g(a)(f(w)) \  (f[r(w)] = r'(f(w))) $,
\item[] (back) $\forall r'\in g(a)(f(w)) \  \exists r\in a(w) \  (f[r(w)] = r'(f(w))) $.
\end{itemize}
It is not hard to see now that the two notions of morphisms are equivalent, when applied to the translated frames respectively.
\end{proof}

\subsection{Proof of Theorem \ref{thm:complet-js} (Completeness of semilattice logic)}\label{ap:semilattices}

\begin{proof}
If $\f \in LSL$, then $\f \in Log(\mathsf{rFSL})$. It is sufficient to show that (\ref{ax-js1}--\ref{ax-js4}) are valid on relational js-frames. This is easily shown using the definition of js-frames. For instance, $w \models \B{0}\f$ iff $\forall r \in \llbracket 0\rrbracket(w) : r(w) \subseteq \llbracket \f\rrbracket$. However, since $\llbracket 0\rrbracket(w) = \emptyset$, this is trivially satisfied for all $w$. As another example, note that $w \models \DB{\p + \pp}\f$ iff there is $r \in \llbracket \p + \pp\rrbracket(w)$ such that $r(w) \subseteq \llbracket \f\rrbracket$. However, $\llbracket \p + \pp\rrbracket(w) = \llbracket\p\rrbracket(w) \cup \llbracket \pp\rrbracket(w)$ and so the previous statement is equivalent to $\exists r \in \llbracket \p\rrbracket(w) : r(w) \subseteq \llbracket \f\rrbracket$ or $\exists r \in \llbracket \pp\rrbracket(w) : r(w) \subseteq \llbracket \f\rrbracket$ which is equivalent to $w \models \DB{\p}\f \lor \DB{\pp}\f$.

\medskip

$\f \in Log(\mathsf{rFSL})$ implies $(\top \approx \f) \in Eq(\mathsf{FSL})$. We reason by contraposition. Fix a js-frame $\mathfrak{A} = (\mathbf{F}, \mathbf{G}, \B{\,}, \DB{\,})$ and an evaluation function $e$ such that $e(\f) \neq e(\top)$ for some $\f \in Fm_{\Sigma}$. We define the relational $\Sigma_{\mathsf{SL}}$-frame $\mathfrak{F} = (\mathrm{Uf}(\mathbf{F}), R, \mathbf{H})$ where $R = (2^{\mathrm{Uf}(\mathbf{F})})^{\mathrm{Uf}(\mathbf{F})}$ and $\mathbf{H} = (H, 0^{\mathbf{H}}, +^{\mathbf{H}})$ is specified as follows: $H = \{ H(\p) \mid \p \in Tm_{\Sigma_{\mathsf{SL}}} \}$ where $H(\p) \in (2^{R})^{\mathrm{Uf}(\mathbf{F})}$ 
such that
\begin{itemize}
\item $H(\mathtt{a})(u) = \{ r_{e(\mathtt{a}), e(\f)} \mid e(\DB{\mathtt{a}}\f) \in u \}$;\footnote{Recall the definition of $r_{a,x}$ in Def.~\ref{def:ultra}: $r_{a,x}(u) = \bigcap \{ \hat{y} \mid \B{a}y \in u \} \cap \hat{x}$.}
\item $H(0)(u) = 0^{\mathbf{H}}(u) = \emptyset$; and
\item $H(\p + \pp)(u) = (H(\p) +^{\mathbf{H}} H(\pp))(u) = H(\p)(u) \cup H(\pp)(u)$.
\end{itemize}
It is clear that $\mathfrak{F}$ is a relational js-frame. We define $V : Tm \cup Fm \to 2^{\mathrm{Uf}(\mathbf{F})} \cup H$ by $V(\chi) = \widehat{e(\chi)}$ and $V(\p) = H(\p)$. We show that $V$ is an interpretation function on $\mathfrak{F}$. $V$ is a Boolean homomorphism by the properties of ultrafilters and it is a $\Sigma$-homomorphism from $Tm$ to $\mathbf{H}$ by the definition of $H(\p)$. (For instance, $V(\p + \pp)(u) = H(\p + \pp)(u) = H(\p)(u) \cup H(\pp)(u) =  (V(\p) +^{\mathbf{H}} V(\pp))(u)$ for all $u$; hence, $V(\p + \pp) = (V(\p) +^{\mathbf{H}} V(\pp)$.) It remains to show the leftmost equalities in the following (recall that we write $H(\gamma, w)$ instead of $H(\gamma)(w)$):
\begin{itemize}
\item[(i)] $V(\B{\ppp}\chi) = \B{V(\ppp)}V(\chi) = \{ w \mid \forall r \in H(\ppp, w) : r(w) \subseteq \widehat{e(\chi)}\}$; and 
\item[(ii)] $V(\DB{\ppp}\chi) = \DB{V(\ppp)}V(\chi) = \{ w \mid \exists r \in H(\ppp, w) : r(w) \subseteq \widehat{e(\chi)}\}$.
\end{itemize}
We show both by induction on the complexity of $\ppp$. (i) The base case $\ppp \in Gr$ is established similarly as the corresponding case in the proof of Lemma \ref{lem:canmorph}, since $H(\ppp)$ is in this case defined as $G(\ppp)$ in the ultrafilter frame. The case $\ppp = 0$ is established as follows: $w \in V(\B{0}\chi)$ iff $e(\B{0}\chi) \in w$ iff $\B{e(0)}e(\chi) \in w$ iff $\B{0} e(\chi) \in w$ iff (using axiom \ref{js1}) $\top \in w$ iff $\forall r \in \emptyset : r(w) \subseteq \widehat{e(\chi)}$ (both are true for all $w$)
iff $\forall r \in H(0, w) : r(w) \subseteq V(\chi)$ iff $w \in \B{V(0)}V(\chi)$. The case $\ppp = \p + \pp$ is established as follows: $w \in V(\B{\p + \pp}\chi)$ iff $w \in e(\B{\p + \pp}\chi)$ iff $w \in \B{e(\p) + e(\pp)}e(\chi)$ iff (using axiom \ref{js3}) $w \in \B{e(\p)}e(\chi)$ and $w \in \B{e(\pp)}e(\chi)$ iff $w \in V(\B{\p}\chi)$ and $w \in V(\B{\pp}\chi)$ iff (by the induction hypothesis) $\forall r \in H(\p, w) : r(w) \subseteq \widehat{e(\chi)}$ and $\forall r \in H(\pp, w) : r(w) \subseteq \widehat{e(\chi)}$ iff $\forall r \in H(\p, w) \cup H(\pp, w) : r(w) \subseteq V(\chi)$ iff $\forall r \in H(\p + \pp, w) : r(w) \subseteq V(\chi)$ iff $w \in \B{V(\p + \pp)}V(\chi)$. Part (ii) is established similarly, using axiom (\ref{js2}) in the case $\ppp = 0$ and axiom (\ref{js4}) in the case $\ppp = \p + \pp$.

Now since $e(\f) \neq e(\top)$, there is $u \in \mathrm{Uf}(\mathbf{F})$ such that $e(\f) \in u$ and $e(\top) \notin u$ by the Prime Filter Theorem. Hence, $V(\f) \neq V(\top)$ and so $\f$ is not valid in the relational js-model $(\mathfrak{F}, V)$. 

\medskip

$(\top \approx \f) \in Eq(\mathsf{FSL})$ implies $\f \in LSL$. 
We define the \emph{canonical $LSL$-frame} as follows.\footnote{The definition of the canonical $LSL$-frame resembles the definition of the basic canonical $L$-frame from Example \ref{ex:canframe-basic}. However, we cannot use $\mathfrak{B}^{LSL}$ here since the group algebra $\mathbf{G}^{LSL}$ in $\mathfrak{B}^{LSL}$ is not a join-semilattice -- it is the term algebra. Hence, we have to define a suitable $LSL$-congruence on the term algebra and prove that it gives rise to a join-semilattice.} Let $\equiv$ be a binary relation on formulas defined as $\equiv_L$ (for $L = LSL$) in Example \ref{ex:canframe-basic}, and let $\equiv^{Tm}$ be a binary relation on $Tm$ such that $\p \equiv^{Tm} \pp$ iff $\B{\p}\f \tot \B{\pp}\f \in LSL$  and $\DB{\p}\f \tot \DB{\pp}\f \in LSL$ for all $\f \in Fm$. Let $[\f]$ be the equivalence class of $\f$ under $\equiv$ and let $[\p]$ be the equivalence class of $\p$ under $\equiv^{Tm}$. It can be shown that $\equiv$ is a congruence on $Fm$ (the usual argument) and $\equiv^{Tm}$ is a congruence on $Tm$. The latter is established using the ``reduction axioms'' for the semilattice operators: if $\B{\p}\f \tot \B{\pp}\f \in LSL$ for all $\f$, then $\B{\p + \ppp}\f \tot \B{\pp + \ppp} \in LSL$ since $\B{\p + \ppp}\f \tot \B{\p}\f \land \B{\ppp}\f \in LSL$ using (\ref{ax-js3}), and so $\B{\p + \ppp}\f \tot \B{\pp}\f \land \B{\ppp}\f \in LSL$ by the assumption which means that $\B{\p + \ppp}\f \tot \B{\pp + \ppp}\f \in LSL$ using (\ref{ax-js3}) again. Hence, we obtain the Boolean algebra $\mathbf{F}$ of equivalence classes $[\f]$, where $o^{\mathbf{F}} ([\f_1], \ldots, [\f_n]) = [o(\f_1, \ldots, \f_n)]$ for all Boolean operators $o$, and the join-semilattice $\mathbf{G}$ of equivalence classes $[\p]$, where $o^{\mathbf{G}} ([\p_1], \ldots, [\p_n]) = [o(\p_1, \ldots, \p_n)]$ for all $o \in \{ 0, + \}$. (The fact that $\mathbf{G}$ is a join-semilattice is easily shown using the reduction axioms: for instance, $[\p + \p] = [\p]$ since $\B{\p + \p}\f \equiv \B{\p}\f$ and $\DB{\p + \p}\f \equiv \DB{\p}\f$ which means that $\pp \in [\p + \p]$ iff $\beta \in [\p]\f$. 
 Moreover, let $\B{\,}$ and $\DB{\,}$ be functions of the type $\mathbf{F} \times \mathbf{G} \to \mathbf{F}$ such that
$$ \B{[\p]} [\f] = [\B{\p}\f] \quad\text{and}\quad
\DB{[\p]} [\f] = [\DB{\p}\f]
$$
(note that these functions are well defined since $\equiv$ and $\equiv^{Tm}$ are both congruences). The \emph{canonical $LSL$-frame} is $\mathfrak{C}^{LSL} = (\mathbf{F}, \mathbf{G}, \B{\,}, \DB{\,})$. It is clear that $\f \in LSL$ iff $\f \approx \top$ is valid in $\mathfrak{C}^{LSL}$ (Prime Filter Theorem). Hence, if $\f \notin LSL$, then there is a js-frame that invalidates $\top \approx \f$, establishing our claim by contraposition.
\end{proof}

\subsection{Proof of Theorem \ref{thm:complet-mon} (Completeness of right-unital magma logic)}

\begin{proof}
$\f \in LRUM$ implies $\f \in Log(\mathsf{rFRUM})$. It is sufficient to check that the extra axioms (\ref{mon1-ax}--\ref{mon4-ax}) of $LRUM$ are valid in all relational rum-frames. The validity of (\ref{mon1-ax}--\ref{mon2-ax}) is clear.  To show that (\ref{mon3-ax}) is valid, we reason as follows: $w \not\models \B{\p}\B{\pp}\f$ iff there are $r \in \llbracket\p\rrbracket(w)$, $u \in r(w)$ and $q \in \llbracket \pp\rrbracket(u)$ such that $q(u) \not\subseteq \llbracket \f\rrbracket$ iff\footnote{Left to right: define $q'$ so that $q'(u) = q(u)$ and $q'(v)$ for $v \neq$ u is fixed in an arbitrary way so that $q' \in \llbracket \pp \rrbracket^{\dagger}$ (this can always be done). Right to left: If $q' \in \llbracket \pp \rrbracket^{\dagger}$ and $q'(u) \neq\emptyset$, then by definition there has to be $q \in \llbracket \pp \rrbracket(u)$ such that $q(u) = q'(u)$.} 
there are $r \in \llbracket\p\rrbracket(w)$, $u \in W$ and $q' \in \llbracket \pp\rrbracket^{\dagger}$ such that $r(w,u)$ and $q'(u) \not\subseteq \llbracket \f\rrbracket$ iff there is $s \in \llbracket \p \cdot \pp\rrbracket(w)$ such that $s(w) \not\subseteq \llbracket \f\rrbracket$ iff $w \not\models \B{\p \cdot \pp} \f$. To show that (\ref{mon4-ax}) is valid, we reason as follows: $w \models \DB{\p \cdot \pp}\f$ iff there is $r \in \llbracket \p\rrbracket(w) \circ \llbracket \pp\rrbracket^{\dagger}$ such that $r(w) \subseteq \llbracket \f\rrbracket$ iff there are $q \in \llbracket \p\rrbracket(w)$ and $s \in \llbracket \pp\rrbracket^{\dagger}$ such that $(q \circ s)(w) \subseteq \llbracket \f\rrbracket$ iff there are $q \in \llbracket \p\rrbracket(w)$ and $s \in \llbracket \pp\rrbracket^{\dagger}$ such that $s(q(w)) \subseteq \llbracket \f\rrbracket$ iff there is $q \in \llbracket\p\rrbracket(w)$ such that for all $u \in q(w)$, if $\llbracket \pp\rrbracket(u) \neq \emptyset$, then there is $t \in \llbracket \pp\rrbracket(u)$ such that $t(u) \subseteq \llbracket \f\rrbracket$ iff $u \models \DB{\p}(\BD{\pp} \bot \lor \DB{\pp}\f)$.

\medskip

$\f \in Log(\mathsf{rFRUM})$ implies $ (\top \approx \f) \in Eq(\mathsf{FRUM})$. We reason by contraposition. Fix a rum-frame $\mathfrak{A} = (\mathbf{F}, \mathbf{G}, \B{\,}, \DB{\,})$ and an evaluation function $e$ such that $e(\f) \neq e(\top)$ for some $\f \in Fm_{\Sigma_{\mathsf{M}}}$. We define the relational $\Sigma_{\mathsf{M}}$-frame $\mathfrak{F} = (\mathrm{Uf}(\mathbf{F}), R, \mathbf{H})$ where $R = (2^{\mathrm{Uf}(\mathbf{F})})^{\mathrm{Uf}(\mathbf{F})}$ and $\mathbf{H} = (H, 1^{\mathbf{H}}, \cdot^{\mathbf{H}})$ is specified as follows: $H = \{ H(\p) \mid \p \in Tm_{\Sigma_{\mathsf{M}}} \}$ where $H(\p) \in (2^{R})^{\mathrm{Uf}(\mathbf{F})}$ such that
\begin{itemize}
\item $H(\mathtt{a})(u) = \{ r_{e(\mathtt{a}), e(\f)} \mid e(\DB{\mathtt{a}}\f) \in u \}$;
\item $H(1)(u) = 1^{\mathbf{H}}(u) =  \{ \mathrm{id}_{\mathrm{Uf}(\mathbf{F})} \}$; and
\item $H(\p \cdot \pp)(u) = (H(\p) \cdot^{\mathbf{H}} H(\pp))(u) =  (H(\p) \otimes H(\pp)(u) = H(\p)(u) \circ H(\pp)^{\dagger}$.
\end{itemize}
It is clear that $\mathfrak{F}$ is a relational rum-frame. We define $V : Tm \cup Fm \to 2^{\mathrm{Uf}(\mathbf{F})} \cup H$ by $V(\chi) = \widehat{e(\chi)}$ and $V(\p) = H(\p)$. We show that $V$ is an interpretation function on $\mathfrak{F}$. $V$ is a Boolean homomorphism by the properties of ultrafilters and it is a $\Sigma_{\mathsf{M}}$-homomorphism from $Tm$ to $\mathbf{H}$ by the definition of $H(\p)$. (For instance, $V(\p \cdot \pp)(u) = H(\p \cdot \pp)(u) = H(\p)(u) \otimes H(\pp)^{\dagger} =  (V(\p) \cdot^{\mathbf{H}} V(\pp))(u)$ for all $u$; hence, $V(\p \cdot \pp) = (V(\p) \cdot^{\mathbf{H}} V(\pp)$.) It remains to establish the leftmost equalities in the following:
\begin{itemize}
\item[(i)] $V(\B{\ppp}\chi) = \B{V(\ppp)}V(\chi) = \{ w \mid \forall r \in H(\ppp, w) : r(w) \subseteq \widehat{e(\chi)}\}$; and 
\item[(ii)] $V(\DB{\ppp}\chi) = \DB{V(\ppp)}V(\chi) = \{ w \mid \exists r \in H(\ppp, w) : r(w) \subseteq \widehat{e(\chi)}\}$.
\end{itemize}
We show both by induction on the complexity of $\ppp$. (i) The base case $\ppp \in Gr$ is established similarly as the corresponding case in the proof of Lemma \ref{lem:canmorph}, since $H(\ppp)$ is in this case defined as $G(\ppp)$ in the ultrafilter frame. The case $\ppp = 1$ is established as follows: $w \in V(\B{1}\chi)$ iff $e (\B{1}\chi) \in w$ iff $\B{e(1)}e(\chi) \in w$ iff $\B{1}e(\chi) \in w$ iff (using axiom \ref{mon1}) $e(\chi) \in w$ iff $\forall r \in H(1, w) : r(w) \subseteq \widehat{e (\chi)}$ iff $\forall r \in H(1, w) : r(w) \subseteq V(\chi)$ iff $w \in \B{V(1)}V(\chi)$. The case $\ppp = \p \cdot \pp$ is established as follows: $w \in V(\B{\p \cdot \pp}\chi)$ iff $\B{e(\p) \cdot e(\pp)}e(\chi) \in w$ iff (using axiom \ref{mon3}) $\B{e(\p)}\B{e(\pp)}e(\chi) \in w$ iff $w \in V(\B{\p}\B{\pp}\chi)$ iff (by induction hypothesis applied to $\p$) $\forall r \in H(\p, w) \forall u (r(w,u) \to u \in V(\B{\pp}\chi))$ iff (by induction hypothesis applied to $\pp$) $\forall r \in H(\p, w) \forall u (r(w,u) \to (\forall q \in H(\pp, u) : q(u) \subseteq V(\chi)))$ iff $\forall r, q (r \in H(\p, w) \And q \in H(\pp)^{\dagger} \to q(r(w)) \subseteq V(\chi))$ iff $\forall r, q (r \in H(\p, w) \And q \in H(\pp)^{\dagger} \to (r \circ q)(w) \subseteq V(\chi))$ iff $\forall s \in H(\p, w) \circ H(\pp)^{\dagger} : s(w) \subseteq V(\chi)$ iff $\forall s \in H(\p \cdot \pp, w) : s(w) \subseteq V(\chi)$ iff $w \in \B{H(\p \cdot \pp)}V(\chi)$ iff $w \in \B{V(\p \cdot \pp)}V(\chi)$. Part (ii) is established similarly, using axiom (\ref{mon2}) in the case $\ppp = 1$ and axiom (\ref{mon4}) in the case $\ppp = \p \cdot \pp$.

\medskip

$ (\top \approx \f) \in Eq(\mathsf{FRUM})$ implies $\f \in LRUM$. 
We define the \emph{canonical $LRUM$-frame} $\mathfrak{C}^{LM}$ similarly as we defined $\mathfrak{C}^{LSL}$ in the third part of the proof of Theorem \ref{thm:complet-js}, but we use $LRUM$ instead of $LSL$, of course. We have to show in our specific setting that (i) $\equiv^{Tm}$ is a congruence and that (ii) $\mathbf{G}$ is a right-unital magma. (i) is established using axioms (\ref{mon3-ax}--\ref{mon4-ax}). To establish (ii), we need to show that $[1]$ is the right unit with respect to $\otimes^{\mathbf{G}}$: to see this, it is sufficient to observe that $\B{\p \cdot 1}\f \equiv \B{\p}\B{1}\f \equiv \B{\p}\f$ and $\DB{\p \cdot 1}\f \equiv \DB{\p}(\BD{1}\bot \lor \DB{1}\f) \equiv \DB{\p}\f$. Hence, $\p \equiv^{Tm} \p \cdot 1$.

As before, the Prime Filter Theorem entails that $\f \in LRUM$ iff $\f \approx \top$ is valid in $\mathfrak{C}^{LRUM}$. Hence, if $\f \notin LRUM$, then there is a rum-frame that invalidates $\top \approx \f$, establishing our claim by contraposition.
\end{proof}

\subsection{Proof of Theorem \ref{thm:complet-cs} (Completeness of closure semilattice logic)}

\begin{proof}
$\f \in LCS$ implies $\f \in Log(\mathsf{rFCS})$. Validity of axioms (\ref{ref-ax}--\ref{cs6-ax}) in relational cs-models is easily checked. 
The rule $\dfrac{\DB{\p}\f \to \DB{\pp}\f}{\DB{\p^{\cap}}\f \to \DB{\pp^{\cap}}\f}$ preserves validity: Assume there is a counterexample to the conclusion of the rule (we may assume $\phi = \texttt{p}$ is atomic). Assume $w \models \DB{\p^{\cap}}\texttt{p}$ and $w \not\models \DB{\pp^{\cap}}\texttt{p}$. Then $w \not\models \DB{\pp} \texttt{p}$. We want to show that there is $\models'$ such that $w \models' \DB{\p}\texttt{p}$ and $w \not\models' \DB{\pp}\texttt{p}$. We know that $\p$ cannot be $0$ since $0^{\cap} = 0$. If $\p = \gamma^{\cap}$ then we are done ($\models'$ is $\models$). Then there is a non-empty $X \subseteq \llbracket \p \rrbracket$(w) such that $\bigcap_{r \in X} r (x) \subseteq \llbracket p \rrbracket$ (W.l.o.g. we consider $X$ to be a minimal nonempty such set). Now define for each $\mathtt{a} \in Gr$ the relation  $r_{\mathtt{a}} = \bigcap \{ r \in X \mid r \in \llbracket \mathtt{a} \rrbracket(w)\}$ and define $\llbracket \mathtt{a} \rrbracket'(w) = \llbracket \mathtt{a} \rrbracket(w) \cup \{ r_{\mathtt{a}} \}$ in case $r_{\mathtt{a}} \neq\emptyset$, and $\llbracket \mathtt{a} \rrbracket'(w) = \llbracket \mathtt{a} \rrbracket(w) \}$ otherwise. The interpretations $\llbracket \gamma \rrbracket'$ of complex $\gamma$ are computed as usual. We can then prove by induction on $\p,\pp$ that $w \models' \DB{\p}\texttt{p}$ while $w\not\models' \DB{\pp}\texttt{p}$. 

\medskip

$\f \in Log(\mathsf{rFCS})$ implies $(\top \approx \f) \in Eq(\mathsf{FCS})$. We reason by contraposition. Fix a cs-frame $\mathfrak{A} = (\mathbf{F}, \mathbf{G}, \B{\,}, \DB{\,})$ and an evaluation function $e$ such that $e(\f) \neq e(\top)$ for some $\f \in Fm_{\Sigma_{\mathsf{M}}}$. Let $\Gamma$ be the smallest set that contains $\f$ and $\top$, is closed under taking subformulas, and
    \begin{itemize}
    \item $\B{\p}\chi \in \Gamma$ iff $\DB{\p}\chi \in \Gamma$
    \item $\B{\p + \pp}\chi \in \Gamma$ only if $\B{\p}\chi \in \Gamma$ and $\B{\pp}\chi \in \Gamma$
    \item $\DB{\p + \pp}\chi \in \Gamma$ only if $\DB{\p}(\neg \DB{\pp}\top \lor \DB{\pp}\chi) \in \Gamma$
    \item $\B{\p^{\cap}}\chi \in \Gamma$ only if $\B{\p}\chi \in \Gamma$
    \item $\DB{\p^{\cap}}\chi \in \Gamma$ only if $\DB{\p}\top \in \Gamma$
    \end{itemize}
    It is easily seen that $\Gamma$ is always finite. 
    For $x \in \mathbf{F}$ and $a \in \mathbf{G}$ such that $x = e (\chi)$ and $a = e(\p)$ for some $\B{\p}\chi \in \Gamma$, we define
  \begin{equation*}
    r^{\Gamma}_{a, x} : w \mapsto \bigcap \{ \widehat{e(\ff)} \mid \B{\p}\f \in \Gamma \And e(\B{\p}\chi) \in w \} \cap \widehat{x} \, .
    \end{equation*}  
   For other pairs of $x, a$ we define $r^{\Gamma}_{a,x} : w \mapsto \emptyset$.
    
  We define the relational frame $\mathfrak{F}^{\Gamma} = (\mathrm{Uf}(\mathbf{F}), R, \mathbf{H})$ as before, but this time $\mathbf{H} = (H, 0^{\mathbf{H}}, +^{\mathbf{H}}, \,^{\cap^{\mathbf{H}}})$ is specified using the relations $r_{a,x}^{\Gamma}$ as follows:
\begin{itemize}
\item $H(\mathtt{a})(u) = \{ r^{\Gamma}_{e(\mathtt{a}), e(\f)} \mid e(\DB{\mathtt{a}}\f) \in u \}$;
\item $H(0)(u) = 0^{\mathbf{H}}(u) = \emptyset$;
\item $H(\p + \pp)(u) = (H(\p) +^{\mathbf{H}} H(\pp))(u) = H(\p)(u) \cup H(\pp)(u)$; and
\item $H(\p^{\cap})(u) = H(\p)^{\cap^{\mathbf{H}}}(u) = \{ r \in R \mid r(u) = \bigcap_{q \in X} q(w) \text{ for some non-empty } X \subseteq H(u)  \}$.
\end{itemize}
It is clear that $\mathfrak{F}$ is a relational cs-frame. We define $V : Tm \cup Fm \to 2^{\mathrm{Uf}(\mathbf{F})} \cup H$ by $V(\chi) = \widehat{e(\chi)}$ and $V(\p) = H(\p)$. We show that $V$ is an interpretation function on $\mathfrak{F}$. $V$ is a Boolean homomorphism by the properties of ultrafilters and it is a $\Sigma$-homomorphism from $Tm$ to $\mathbf{H}$ by the definition of $H(\p)$. We would like to establish the leftmost equalities in the following:
\begin{itemize}
\item[(i)] If $\B{\ppp}\chi \in \Gamma$, then $V(\B{\ppp}\chi) = \B{V(\ppp)}V(\chi) = \{ w \mid \forall r \in H(\ppp, w) : r(w) \subseteq \widehat{e(\chi)}\}$; and 
\item[(ii)] if $\B{\ppp}\chi \in \Gamma$, then $V(\DB{\ppp}\chi) = \DB{V(\ppp)}V(\chi) = \{ w \mid \exists r \in H(\ppp, w) : r(w) \subseteq \widehat{e(\chi)}\}$.
\end{itemize}
(i) can be show easily by induction on the complexity of $\ppp$. However, similarly to the extensional case, only the $\supseteq$ inclusion of (ii) can be shown to hold for $\mathfrak{F}^{\Gamma}$. Here we can use the standard technique of splitting to transform $\mathfrak{F}^{\Gamma}$ into a correct relational cs-frame (see the extended version of \cite{BilkovaEtAl2021}). We omit the details.

\medskip

$(\top \approx \f) \in Eq(\mathsf{FCS})$ implies $\f \in LCS$. 
We define the \emph{canonical $LCS$-frame} $\mathfrak{C}^{LCS}$ similarly as we defined $\mathfrak{C}^{LSL}$ in the third part of the proof of Theorem \ref{thm:complet-js}, but we use $LCS$ instead of $LSL$, of course. We have to show in our specific setting that (i) $\equiv^{Tm}$ is a congruence and that (ii) $\mathbf{G}$ is a closure semilattice. Both items are checked using the corresponding closure axioms, and the rule (to show that $^{\cap}$ is a monotonic operator on the $\equiv^{Tm}$-quotient algebra).

\end{proof}

\end{document}